\documentclass[10pt,conference]{IEEEtran}
\usepackage{fixltx2e}
\usepackage{cite}
\usepackage{url}
\usepackage{mathrsfs}
\usepackage{color}
\usepackage{float}
\usepackage[caption=false]{subfig}

\ifCLASSINFOpdf
   \usepackage[pdftex]{graphicx}
   \graphicspath{{Figs/}}
   \DeclareGraphicsExtensions{.pdf,.jpeg,.png}
\else
\fi

\usepackage[cmex10]{amsmath}
\usepackage{amsmath}
\usepackage{amssymb}
\usepackage{amsthm}
\usepackage{amsfonts}
\usepackage{bm}
\usepackage{xfrac}
\usepackage{empheq}
\usepackage[normalem]{ulem} % either use this (simple) or
\usepackage{soul} % use this (many fancier options)
\usepackage{mathtools}
\DeclarePairedDelimiter\ceil{\lceil}{\rceil}
\DeclarePairedDelimiter\floor{\lfloor}{\rfloor}

\newtheorem{theorem}{Theorem}

\newtheorem{example}{Example}

\newtheorem{lemma}{Lemma}

\newtheorem{remark}{Remark}
\theoremstyle{definition}

\usepackage[a4paper,bindingoffset=0.2in,%
left=1in,right=1in,top=1in,bottom=1in,%
footskip=.25in]{geometry}
\graphicspath{{figs/}}
\newenvironment{sproof}{\noindent{ \emph{ Sketch of proof:}}}{\qed\bigskip}
\begin{document}
	\newgeometry{left=0.7in,right=0.7in,top=.7in,bottom=1in}
	\title{Private Variable-Length Coding with Sequential Encoder}
	%\title{Bounds for Multi-User Privacy-Utility Trade-off with Non-zero Leakage}
\vspace{-5mm}
\author{
		\IEEEauthorblockN{Amirreza Zamani$^\dagger$, Tobias J. Oechtering$^\dagger$, Deniz G\"{u}nd\"{u}z$^\ddagger$, Mikael Skoglund$^\dagger$ \vspace*{0.5em}
			\IEEEauthorblockA{\\
                              $^\dagger$Division of Information Science and Engineering, KTH Royal Institute of Technology \\
                              $^\ddagger$Dept. of Electrical and Electronic Engineering, Imperial College London\\
				Email: \protect amizam@kth.se, oech@kth.se, d.gunduz@imperial.ac.uk, skoglund@kth.se }}%\vspace*{-2.1em}
		}
	\maketitle
%
%\iffalse
\begin{abstract} 
	A multi-user private data compression problem is studied. A server has access to a database of $N$ files, $(Y_1,...,Y_N)$, each of size $F$ bits and is connected 
	to an encoder. The encoder is connected
	through an unsecured link to a user. 
	%In the placement phase, the server and encoder fill the users' caches without knowing their demands, while the delivery phase takes place after the users send their demands to the server. 
	We assume that each file $Y_i$ is arbitrarily correlated with a private attribute $X$, which is assumed to be accessible by the encoder.
	Moreover, an adversary is assumed to have access to the link. The users and the encoder have access to a shared secret key $W$. We assume that at each time the user asks for a file $Y_{d_i}$, where $(d_1,\ldots,d_K)$ corresponds to the demand vector.  
	The goal is to design the delivered message $\mathcal {C}=(\mathcal {C}_1,\ldots,\mathcal {C}_K)$ after the user send his demands to the encoder
	such that the average length of $\mathcal{C}$ is minimized, while satisfying: 
i. The message $\cal C$ does not reveal any information about $X$, i.e., $X$ and $\mathcal{C}$ are independent, which corresponds to the perfect privacy constraint; ii. The user is able to decode its demands, $Y_{d_i}$, by using $\cal C$, and the shared key $W$.
Here, the encoder sequentially encode each demand $Y_{d_i}$ at time $i$, using the shared key and previous encoded messages.  
 We propose a variable-length coding scheme that uses privacy-aware compression techniques. We study proposed upper and lower bounds on the average length of $\mathcal{C}$ in an example. Finally, we study an application considering cache-aided networks. 
	% Since the database is correlated with $X$, existing codes for cache-aided delivery do not satisfy the perfect privacy condition. Indeed, we propose a variable-length coding scheme that combines privacy-aware compression with coded caching techniques. In particular, we use two-part code construction and Functional Representation Lemma. 
	%Finally, we extend the results to the case, where $X$ and $\mathcal{C}$ can be correlated, i.e., non-zero leakage is allowed. 
\end{abstract}
\section{Introduction}
 We consider the scenario illustrated in Fig. \ref{wii}, in which a server has access to a database consisting of $N$ files $Y_1,\ldots,Y_N$, where each file, of size $F$ bits, is sampled from the joint distribution $P_{XY_1\cdot Y_N}$, where $X$ denotes the private latent variable, whose realization is known to the server. The user requests $K\leq N$ files from the server sequentially, where $d_i\in[N]\triangleq\{1,\ldots,N\}$ represents the user request at time $i$, $d_i\neq d_j$ for $i\neq j$. We assume that the server delivers the user's request over an unsecured link, but we assume that the two have access to a shared secret key denoted by $W$, of size $T$. %The system works in two phases: the placement and delivery phases, respectively, \cite{maddah1}. In the placement phase, the server and encoder fill the local caches using the database. Let $Z_k$ denote the content of the local cache memory of user $k$, $k\in[K]\triangleq\{1,..,K\}$ after the placement phase. In the delivery phase, 
 In this work, the user send his demands to the encoder, where $d_i\in[N]\triangleq\{1,\ldots,N\}$ denotes the demand of the user at time $i$. We assume that the user asks for $K$ files where $K\leq N$ and $d_i\neq d_j$ for $i\neq j$.
 At time slot $i$, the encoder receives $d_i$ and designs a message $\mathcal{C}_i$ using $W$ and the previous messages it has delivered $(\mathcal{C}_1,\mathcal{C}_2,\ldots,\mathcal{C}_{i-1})$ to satisfy $d_i$. We note that the encoder is not aware of the future demands. 
 %The encoder sends a response, denoted by $\mathcal{C}=(\mathcal{C}_1,\mathcal{C}_2,\ldots,\mathcal{C}_K)$, $\mathcal{C}_i$ at $i$-th channel use, to satisfy all the demands. %The encoder designs a message $\mathcal{C}_i$ using the $i$-th block of the server's response, i.e., $\mathcal{C}_i'$, and the shared key and send it over the shared link to the users. It can also store $\mathcal{C}_i$ in the public cache which is assumed to be not accessible by the users. The public cache enables the encoder to use previous stored messages to design the next message, i.e., the encoder can use $(\mathcal{C}_1,\ldots,\mathcal{C}_{i-1})$ to design message $\mathcal{C}_i$.
 
 We assume that the delivery channel is public, and an adversary can use the delivered messages $\mathcal{C}=(\mathcal{C}_1,\mathcal{C}_2,\ldots,\mathcal{C}_K)$ to extract information about $X$. %However, the adversary does not have access to the users' local cache contents or the secret key. %, i.e., it only listens to the channel and receives the response $\cal C$. 
 %Since the files in the database are all correlated with the private latent variable $X$, the coded caching and delivery techniques introduced in \cite{maddah1} do not satisfy the privacy requirement. 
 The goal of the private delivery problem is to find a response $\mathcal{C}$ with minimum average length that satisfies user's demands while guaranteeing a certain privacy requirement. %Furthermore, the contents of the public cache $P$ must satisfy the privacy constraint. 
 Here, we consider the worst case demand combinations $d=(d_1,..,d_K)$ to construct $\cal C$, and the expectation is taken over the randomness in the database. In this work, we impose a perfect privacy constraint, i.e., we require $\mathcal {C}$ to be independent of $X$. Let $\hat{Y}_{d_i}$ denote the decoded message of the user at time slot $i$ using $W$ and $(\mathcal{C}_1,\mathcal{C}_2,\ldots,\mathcal{C}_{i})$. The user should be able to recover $Y_{d_k}$ reliably, i.e., $\mathbb{P}{\{\hat{Y}_{d_k}\neq Y_{d_k}\}}=0$, $\forall k\in[K]$. 
 
 We have a variable-length compression problem with a perfect privacy constraint. To solve this problem  % We use zero error probability constraint for each user, i.e., $\mathbb{P}{\{\hat{Y}_{d_i}\neq Y_{d_i}\}}=0$.
  we combine techniques used in privacy mechanisms \cite{king1} and data compression \cite{kostala}. %and combine them to build such a code. %Here, we introduce an approach and apply it to a lossless data compression problem in a cache-aided network. To do so, 
\begin{figure}[]
	\centering
	\includegraphics[scale = .1]{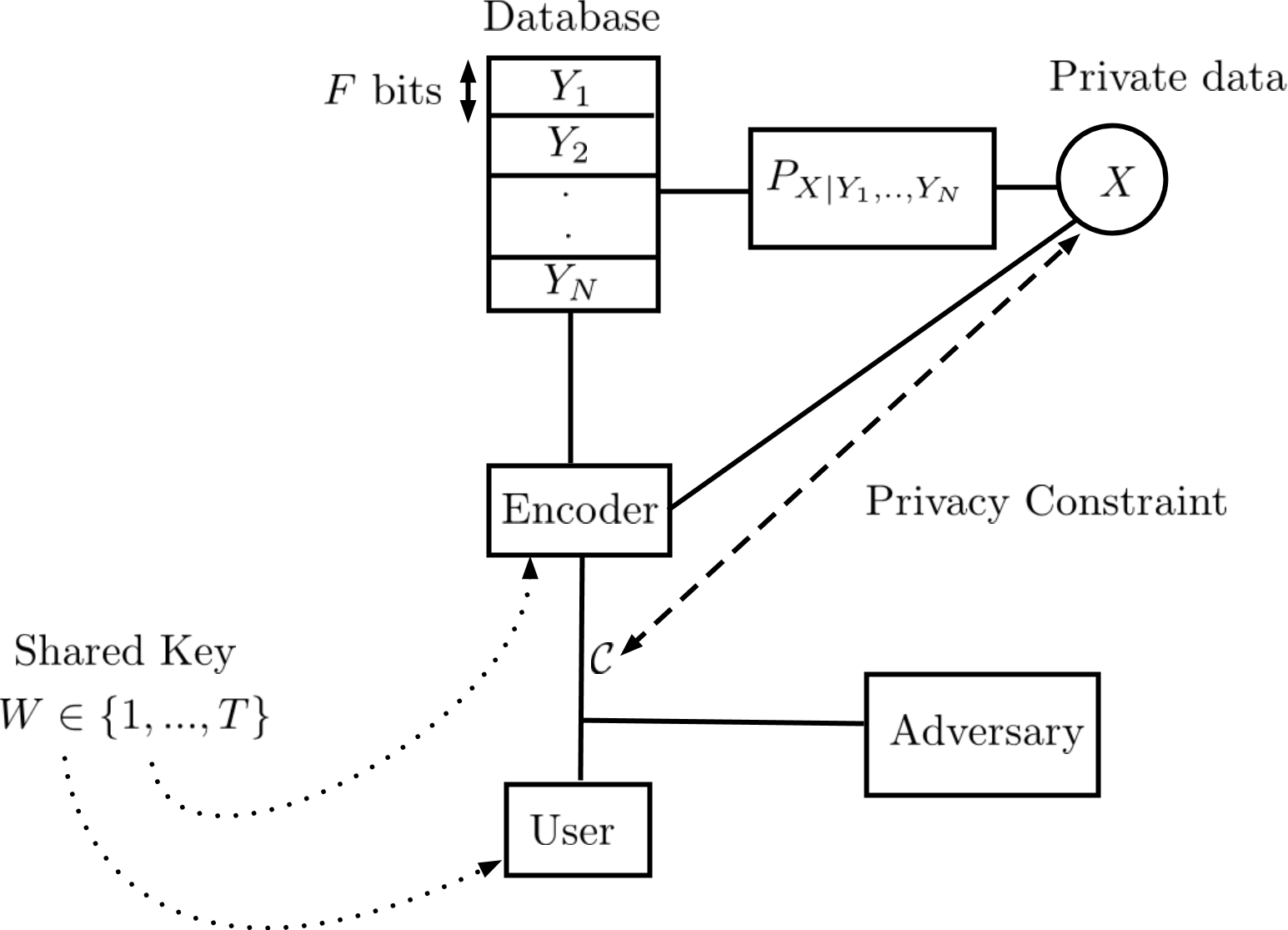}
	\caption{A server sequentially delivers requests of a user from its database over an unsecured public channel.}
	%In this work an encoder wants to compress $Y$ which is correlated with $X$ under certain privacy leakage constraints and send it over a channel where an eavesdropper has access to the output of the encoder. The encoder and decoder have an advantage of using shared secret key.}
	\label{wii}
\end{figure}
Privacy mechanisms and compression problems have received increasing attention in recent years
%The privacy mechanism design problem is receiving increased attention in information theory recently. 
%Related works can be found in 
\cite{shannon, denizjadid2, gunduz2010source, schaefer, sankar, yamamoto1988rate, yeye,yaya, makhdoumi, cufff, borz, khodam, Khodam22,kostala, kostala2, calmon4, asoo, issa2, king1, king2, zamani2023cache,asoodeh1,king3,deniz3}. 
%Specifically, in \cite{maddah1}, a cache-aided network consisting of a single server connected
%to several users through a shared bottleneck error-free link, is
%considered, and the rate-memory trade-off has been characterized within a constant gap. This bound has been improved
%for different scenarios in \cite{lim,wang33}. In particular, the exact rate-memory trade-off for uncoded placement has been characterized in \cite{lu}. 
%In \cite{denizjadid}, a cache-aided coded content delivery problem is studied where users have different distortion requirements.
%In \cite{gholami}, the problem of coded caching with private demands and caches has been studied, where the private coded caching schemes are constructed using private information retrieval (PIR).
A notion of perfect secrecy was introduced in \cite{shannon} by Shannon, where the public and private data are statistically independent. %In the Shannon cipher system, one of $M$ messages is sent over a channel wiretapped by an eavesdropper, and it is shown that perfect secrecy is achievable if and only if the shared secret key length is at least $M$ \cite{shannon}. 
%Perfect secrecy, where public data and private data are independent, is also considered in \cite{dworkal,dwork1}, where sanitized version of a database is disclosed for public use. 
Equivocation as a measure of information leakage for information theoretic security has been used in \cite{denizjadid2,gunduz2010source, schaefer, sankar}. A rate-distortion approach to information theoretic secrecy is studied in \cite{yamamoto1988rate}. 
Lossless data compression with secrecy constraints are studied in \cite{yeye,yaya}.
%A related source coding problem with secrecy is studied in \cite{yamamoto}.  
%Fundamental limits of the privacy utility trade-off measuring the leakage using estimation-theoretic guarantees are studied in \cite{Calmon2}.
%The concept of maximal leakage has been introduced in \cite{issa} and some bounds on the privacy utility trade-off have been derived. 
The concept of privacy funnel is introduced in \cite{makhdoumi}, where the privacy-utility trade-off has been studied considering the log-loss as a privacy measure as well as a distortion measure for utility. %and the log-loss as 
%The privacy-utility trade-offs considering equivocation and expected distortion as measures of privacy and utility are studied in both \cite{sankar} and \cite{yamamoto}.
In \cite{borz}, the problem of privacy-utility trade-off considering mutual information both as measures of utility and privacy is studied. It is shown that under the perfect privacy assumption, the optimal privacy mechanism can be obtained as the solution of a linear program. %This has been extended in \cite{gun} considering the privacy utility trade-off with a rate constraint on the disclosed data.
%Moreover, in \cite{borz}, it has been shown that information can be only revealed if the kernel (leakage matrix) between useful data and private data is not invertible.
The result in \cite{borz} is generalized to some small bounded leakage in \cite{khodam}, and %More specifically, privacy mechanisms with a per-letter privacy criterion considering an invertible kernel are designed allowing a small leakage. 
to a non-invertible leakage matrix in \cite{Khodam22}.
In \cite{kostala}, an approach to partial secrecy called \emph{secrecy by design} is introduced, and applied to two problems: privacy mechanism design and lossless compression. For the privacy design problem, bounds on privacy-utility trade-off are derived by using the Functional Representation Lemma (FRL). These results are derived under the perfect secrecy assumption.
%In \cite{kostala}, the problem of \emph{secrecy by design} is studied and bounds on privacy-utility trade-off for two scenarios where the private data is hidden or observable are derived by using the Functional Representation Lemma. These results are derived under the perfect secrecy assumption, i.e., no leakages are allowed. %The bounds are tight when the private data is a deterministic function of the useful data.
In \cite{king1}, the privacy problems considered in \cite{kostala} are generalized by relaxing the perfect secrecy requirement. %More specifically, we considered bounded mutual information, i.e., $I(U;X)\leq \epsilon$ for privacy leakage constraint.
%Furthermore, in the special case of perfect privacy we derived a new upper bound for the perfect privacy function and it has been shown that this new bound generalizes the bound in \cite{kostala}. Moreover, it has been shown that the bound is tight when $|\mathcal{X}|=2$.\\
In \cite{king2}, the privacy-utility trade-off with two different per-letter privacy constraints is studied. %Upper and lower bounds are derived and it has been shown that the bounds in the first scenario, where the private data is hidden to the agent, are asymptotically optimal when the private data is a deterministic function of useful data. \\
Moreover, in \cite{kostala}, both fixed-length and variable-length compression problems have been studied, and upper and lower bounds on the average length of the encoded message have been derived. These results are derived under the assumption that the private data is independent of the encoded message.
A similar compression problem in cache-aided networks is studied in \cite{zamani2023cache} in the presence of an adversary.  %A similar approach has been used in \cite{kostala2}, where in a lossless compression problem the relations between secrecy, shared key and compression considering perfect secrecy, secrecy by design, maximal leakage, mutual information leakage and local differential privacy have been studied.
%Our problem here is closely related to \cite{zamani2023cache} and \cite{kostala}, where in \cite{zamani2023cache} a private cache-aided data delivery problem in the presence of an adversary is considered.  %where in \cite{kostala}, for the problem of lossless data compression, strong information theoretic guarantees are provided and fundamental limits are characterized when the private data is a deterministic function of the useful data. 
In this paper, in contrast to \cite{zamani2023cache}, we consider a sequential encoder, and
we use
variable-length lossless compression techniques as in \cite{kostala} to find an alternative solution to data delivery in the presence of an adversary. Our key contribution is a multi-part code construction based on an extension of FRL combined with a one-time-pad scheme. %Finally, we generalize the results for non-zero leakage, i.e., $X$ and $\mathcal{C}$ are allowed to be correlated.  

\section{system model and Problem Formulation} \label{sec:system}
Let $Y_i$ denote the $i$-th file in the database, which is of size $F$ bits, i.e., $\mathcal{Y}_i\in\{1,\ldots,2^F\}$ and $|\mathcal{Y}_i|=2^F$. %In this work we assume that $N\geq K$; however, the results can be generalized to other cases as well. 
Let the discrete random variable (RV) $X$ defined on the finite alphabet $\cal{X}$ describe the private latent variable and be arbitrarily correlated with the files in the database $Y=(Y_1,\ldots,Y_N)$ where $|\mathcal{Y}|= |\mathcal{Y}_1|\times\ldots\times|\mathcal{Y}_N|=(2^F)^N$ and $\mathcal{Y}= \mathcal{Y}_1\times\ldots\times\mathcal{Y}_N$. %$|\mathcal{Y}|= (2^F)^N$. 
We denote the joint distribution of the private data and the database by $P_{XY_1\cdot Y_N}$ and marginal distributions of $X$ and $Y_i$ by vectors $P_X$ and $P_{Y_i}$ defined on $\mathbb{R}^{|\mathcal{X}|}$ and $\mathbb{R}^{2^F}$. % given by the row and column sums of $P_{XY_1\cdot Y_N}$. 
The relation between $X$ and $Y$ is given by the matrix $P_{X|Y_1\cdot Y_N}$ defined on $\mathbb{R}^{|\mathcal{X}|\times(2^F)^N}$. %Let $P$ denotes the content of the public cache.
The shared secret key is denoted by the discrete RV $W$ defined on $\{1,\ldots,T\}$, and is assumed to be known by the encoder and the users, but not the adversary. Furthermore, we assume that $W$ is uniformly distributed and is independent of $X$ and $Y$. %Let $Y=(Y_1,...,Y_N)$ where $|\mathcal{Y}|\leq (2^F)^N$ and let $[K]=\{1,..,K\}$. 
%Similarly to \cite{maddah1}, we have $K$ caching functions to be used during the placement phase:
%\begin{align}
%\theta_k: [|\mathcal{Y}|] \rightarrow [2^{\floor{FM}}],\ \forall k\in[K], 
%\end{align} 
%such that
%\begin{align}
%Z_k=\theta_k(Y_1,\ldots,Y_N),\ \forall k\in[K].
%\end{align} 
Let the vector $(Y_{d_1},\ldots,Y_{d_K})$ denote the demands of the user, where ${d_i}$ is sent at time $i$ and $(d_1,\ldots,d_K)\in[N]^K$.  
A variable-length prefix-free code with a shared secret key of size $T$ consists of mappings:
\begin{align*}
%&(\text{server})\! \ \mathcal{C}': [|\mathcal{Y}|]\times[N]^K \rightarrow \{0,1\}^*\\
&(\text{encoder})\! \ \mathcal{C}_1: [T]\!\times\![N]\! \rightarrow\! \{0,\!1\}^*\\
&(\text{encoder})\! \ \mathcal{C}_i: \{0,\!1\}^*\!\times\!\! [T]\!\times\!\![N]^i\! \rightarrow\! \{0,\!1\}^*,\ \ i\geq2,\\
&(\text{decoder}) \mathcal{D}_i\!: \! \{0,1\}^*\!\times\! [T]\!\times\! [N]^i\!\rightarrow\! 2^F\!,\ \! \forall i\!\in\![K].
\end{align*}
The output of the encoder in the $i$-th time slot $\mathcal{C}(Y_{d_i},d_1,\ldots,d_i,\mathcal{C}_1,\ldots, \mathcal{C}_{i-1})$ is the codeword the server sends over the link to satisfy the demands of the user $(Y_{d_1},\ldots,Y_{d_i})$. %Let $\mathcal{C}'_i$ be the i-th block of the codeword $\mathcal {C}'$. 
%Due to the limited size of the local buffer $C$, The encoder uses a sequentially coding scheme as follows: 
%First, the encoder receives $\mathcal{C}'_1$ and encode it using the shared key. Let the output be $\mathcal{C}_1$. The encoder sends $\mathcal{C}_1$ over the shared link and also stores it to the public cache. Then, the encoder receives the second block $\mathcal{C}'_2$ and encode it using $W$ and $\mathcal{C}_1$. The output is denoted by $\mathcal{C}_2$ and is sent through the shared link and is stored to the public cache.  
%Here, we assume that $L$ is a large quantity. Furthermore, the encoder encodes i-th block $\mathcal{C}'_i$ using $W$ and $(\mathcal{C}'_1,\ldots,\mathcal{C}'_{i-1})$. The encoding scheme continues until the last block of the codeword $\mathcal {C}'$. Let $\mathcal{C}=(\mathcal{C}_1,\mathcal{C}_2,\ldots)$ denote the delivered message consisting of the messages sent over the shared link.  
%The output of the encoder $\mathcal{C}(Y,W,d_1,\ldots,d_K)$ is the codeword the server sends over the shared link in order to satisfy the demands of the users $(Y_{d_1},\ldots,Y_{d_K})$. 
At the user side, the user employs the decoding function $\mathcal{D}_i$ in the $i$-th time slot to recover its demand $Y_{d_i}$, i.e., $\hat{Y}_{d_i}=\mathcal{D}_i(W,d_1,\ldots,d_i,\mathcal{C}_1,\ldots, \mathcal{C}_{i-1})$.
Since the code is prefix free, no codeword in the image of $\cal C$ is a prefix of another codeword. The variable-length code $(\mathcal{C}_1,\ldots,\mathcal{C}_K,\mathcal{D}_1,\ldots,\mathcal{D}_K)$ is lossless if for all $k\in[K]$ we have
\begin{align}\label{choon}
\mathbb{P}(\mathcal{D}_i(W,d_1,\ldots,d_i,\mathcal{C}_1,\ldots, \mathcal{C}_{i-1})\!=\!Y_{d_i})\!=\!1.
\end{align} 
In the following, we define perfectly private codes. %and $\epsilon$-private codes.
The code $(\mathcal{C}_1,\ldots,\mathcal{C}_K,\mathcal{D}_1,\ldots,\mathcal{D}_K)=(\mathcal{C},\mathcal{D}_1,\ldots,\mathcal{D}_K)$ is \textit{perfectly private} if
\begin{align}
I(\mathcal{C};X)=0.\label{lashwi}
\end{align}
%and the code is 
%\textit{$\epsilon$-private} if
%\begin{align}
%I(\mathcal{C};X)=\epsilon.\label{lash1}
%\end{align}
Let $\xi$ be the support of $\mathcal{C}$, where $\xi\subseteq \{0,1\}^*$. For any $c\in\xi$, let $\mathbb{L}(c)$ be the length of the codeword. The lossless code $(\mathcal{C},\mathcal{D}_1,\ldots,\mathcal{D}_K)$ is \textit{$(\alpha,T,d_1,\ldots,d_K)$-variable-length} if 
\begin{align}\label{jojowi}
\mathbb{E}(\mathbb{L}(\mathcal{C}))\!\leq\! \alpha,\ \forall w\!\in\!\![T]\ \text{and}\ \forall d_1,\ldots,d_K,
\end{align} 
and $(\mathcal{C},\mathcal{D}_1,\ldots,\mathcal{D}_K)$ satisfies \eqref{choon}.
Finally, let us define the set $\mathcal{H}(\alpha,T,d_1,\ldots,d_K)$ as follows:\\ %and $\mathcal{H}^{\epsilon}(\alpha,T,d_1,\ldots,d_K)$ as follows:\\
$\mathcal{H}(\alpha,T,d_1,\ldots,d_K)\triangleq\{(\mathcal{C},\mathcal{D}_1,\ldots,\mathcal{D}_K): (\mathcal{C},\mathcal{D}_1,\ldots,\mathcal{D}_K)\ \text{is}\ \text{perfectly-private and}\\ (\alpha,T,d_1,\ldots,d_K)\text{-variable-length}  \}$. % and $\mathcal{H}^{\epsilon}(\alpha,T,d_1,\ldots,d_K)\triangleq\{(\mathcal{C}',\mathcal{C},\mathcal{D}_1,\ldots,\mathcal{D}_K): (\mathcal{C}',\mathcal{C},\mathcal{D}_1,\ldots,\mathcal{D}_K)\ \text{is}\ \epsilon\text{-private and}\\ (\alpha,T,d_1,\ldots,d_K)\text{-variable-length}  \}.$
The private compression design problem with sequential encoder can be stated as follows
\begin{align}
\mathbb{L}(P_{XY_1\cdot Y_N},T)&=\!\!\!\!\!\inf_{\begin{array}{c} 
	\substack{(\mathcal{C},\mathcal{D}_1,\ldots,\mathcal{D}_K)\in\mathcal{H}(\alpha,T,d_1,\ldots,d_K)}
	\end{array}}\alpha.\label{main1wi}%\\
%\mathbb{L}^{\epsilon}(P_{XY_1\cdot Y_N},T,C)&=\!\!\!\!\!\!\!\!\!\!\inf_{\begin{array}{c} 
%	\substack{(\mathcal{C}',\mathcal{C},\mathcal{D}_1,\ldots,\mathcal{D}_K)\in\mathcal{H}^{\epsilon}(\alpha,T,d_1,\ldots,d_K)}
%	\end{array}}\alpha.\label{main2wi}
\end{align} 
\begin{remark}
	\normalfont 
	Letting $N=K=1$, \eqref{main1wi} leads to the privacy-compression rate trade-off studied in \cite{kostala}, where upper and lower bounds have been derived.
	%statistically independence between the encoded message and the private data, both \eqref{main1} and \eqref{main3} lead to the privacy-compression rate trade-off studied in \cite{kostala}, where upper and lower bounds have been derived. In this paper, we generalize the trade-off by considering non-zero $\epsilon$.
\end{remark}
%\begin{remark}
%	\normalfont 
%	By letting $M=0$, and considering $X=(X_1,\ldots,X_N)$, a similar privacy-utility trade-off has been studied in \cite{zamani2022multi}, where each user wants a subvector of the database $Y=(Y_1,\ldots,Y_N)$ that is correlated with $X=(X_1,\ldots,X_N)$ and the server maximizes a linear combination of utilities. Using the Functional Representation Lemma and Strong Functional Representation Lemma, upper and lower bounds have been derived, which are shown to be tight within a constant. 
	%statistically independence between the encoded message and the private data, both \eqref{main1} and \eqref{main3} lead to the privacy-compression rate trade-off studied in \cite{kostala}, where upper and lower bounds have been derived. In this paper, we generalize the trade-off by considering non-zero $\epsilon$.
%\end{remark}
\begin{remark}
	In this paper, to design a code, we consider the worst case demand combinations. This follows since \eqref{jojowi} must hold for all possible combinations of the demands.
\end{remark}

 \section{Main Results}\label{sec:resul}
In this section, we derive an upper bound on $\mathbb{L}(P_{XY_1\cdot Y_N},T)$ %and $\mathbb{L}^{\epsilon}(P_{XY_1\cdot Y_N},T)$ 
defined in \eqref{main1wi}. %and \eqref{main2wi}. 
For this, we employ the multi-part code construction, which is similar to the two-part code used in \cite{kostala}. We first encode the private data $X$ using a one-time-pad \cite[Lemma~1]{kostala2}, then sequentially encode the user's demands using an extension of the Functional Representation Lemma. To do so, let us, first recall FRL.
\begin{lemma}\label{FRL}(FRL \cite[Lemma~1]{kostala}):
	For any pair of RVs $(X,Y)$ distributed according to $P_{XY}$ supported on alphabets $\mathcal{X}$ and $\mathcal{Y}$, respectively, where $|\mathcal{X}|$ is finite and $|\mathcal{Y}|$ is finite or countably infinite, there exists a RV $U$ supported on $\mathcal{U}$ such that $X$ and $U$ are independent, i.e., 
	$
	I(U;X)= 0,
	$
	$Y$ is a deterministic function of $U$ and $X$, i.e., 
	$
	H(Y|U,X)=0,
	$
	and 
	$
	|\mathcal{U}|\leq |\mathcal{X}|(|\mathcal{Y}|-1)+1.
	$
	Furthermore, if $X$ is a deterministic function of $Y$, we have
	$
	|\mathcal{U}|\leq |\mathcal{Y}|-|\mathcal{X}|+1.
	$
\end{lemma}  
The proof of Lemma~\ref{FRL} is constructive and is useful to obtain the next lemma.
Next, we provide an extension of FRL %from \cite{kostala}, which shows that there exists a RV $U$ that satisfies \eqref{t1wi} and \eqref{t2wi}, and has bounded entropy. Lemma~\ref{aghabwi} has a constructive proof as well, that can 
that helps us find upper bound on $\mathbb{L}(P_{XY_1\cdot Y_N},T)$. %and $\mathbb{L}^{\epsilon}(P_{XY_1\cdot Y_N},T)$. 

%\begin{lemma}\label{aghabwi} (\cite[Lemma~2]{kostala})
%	For any pair of RVs $(X,Y)$ distributed according to $P_{XY}$ supported on alphabets $\mathcal{X}$ and $\mathcal{Y}$, respectively, where $|\mathcal{X}|$ is finite and $|\mathcal{Y}|$ is finite or countably infinite, there exists a RV $U$ such that it satisfies \eqref{t1wi}, \eqref{t2wi}, and
%	\begin{align}
%	H(U)\leq \sum_{x\in\mathcal{X}}H(Y|X=x).\label{wi}
%	\end{align}
%\end{lemma}
\begin{lemma}\label{chaghal}
	For a fixed integer $k\geq 1$, let RVs $(X,Y,U_1,\ldots U_k)$ be distributed according to $P_{XYU_1\cdot U_k}$ where $I(X;U_1,\ldots U_k)=0$. Then, there exists a RV $U_{k+1}$ such that $X$ and $(U_1,\ldots,U_{k+1})$ are independent, i.e., 
	\begin{align}\label{t12}
	I(U_1,\ldots,U_{k+1};X)= 0,
	\end{align}
	$Y$ is a deterministic function of $(U_1,\ldots,U_{k+1})$ and $X$, i.e., 
	\begin{align}\label{t22}
	H(Y|U_1,\ldots,U_{k+1},X)=0,
	\end{align}
	and 
	\begin{align}\label{pr2}
	|\mathcal{U}_{k+1}|\leq |\mathcal{X}||\mathcal{U}_1|\cdot|\mathcal{U}_{k}|(|\mathcal{Y}|-1)+1.
	\end{align}
	%Furthermore, if $X$ is a deterministic function of $Y$, we have
	%\begin{align}\label{prove2}
	%|\mathcal{U}|\leq |\mathcal{Y}|-|\mathcal{X}|+1.
	%\end{align}
\end{lemma}
\begin{proof}
	Let $U_1^{i}\triangleq(U_1,\ldots,U_i)$ and let $U_{k+1}$ be produced based on FRL using $X\leftarrow (X,U_1,\ldots,U_k)$ and $Y\leftarrow Y$. We have
	\begin{align*}
	I(X;U_1^{k+1})=I(X;U_1^{k})+I(X;U_{k+1}|U_1^k)=0,
	\end{align*}
	where the last line follows since the first term is zero by the assumption and the second term is zero due to independence of $U_{k+1}$ and $(X,U_1,\ldots,U_k)$.
\end{proof}
Next, we present our achievable scheme which leads to an upper bound on \eqref{main1wi}. Let $m_1=\min_{\begin{array}{c} 
	\substack{U_1:I(U_1;X)=0,\\H(Y_{d_1}|X,U_1)=0}
	\end{array}}H(U_1)$ and $U_1^*$ be an optimizer that achieves $m_1$. Furthermore, let $m_2=\min_{\begin{array}{c} 
	\substack{U_2:I(U_2,U_1^*;X)=0,\\H(Y_{d_2}|X,U_1^*,U_2)=0}
	\end{array}}H(U_2)$ and $U_2^*$ be an optimizer that achieves $m_2$. Similarly, let $m_i=\min_{\begin{array}{c} 
	\substack{U_i:I(U_1^*,\ldots,U_{i-1}^*,U_i;X)=0,\\H(Y_{d_i}|X,U_1^*,\ldots,U_{i-1}^*,U_i)=0}
	\end{array}}H(U_i)$ and $U_i^*$ be an optimizer that achieves $m_i$. %For simplicity let % $Y=(Y_1,..,Y_N)$ where $|\mathcal{Y}|$ is the size of database where $|\mathcal{Y}|\leq |\mathcal{Y}_1|\times..\times|\mathcal{Y}_N|$ and $\mathcal{Y}\subseteq \mathcal{Y}_1\times..\times\mathcal{Y}_N$. Furthermore, 
%$|\mathcal{C}'|$ be the cardinality of the codeword defined in \eqref{cache1} where $|\mathcal{C}'|= |\mathcal{C}_{\gamma_1}|\times\ldots\times|\mathcal{C}_{\gamma_Q}|$. 
\begin{theorem}\label{th1}
	Let RVs $(X,Y)=(X,Y_1,\ldots,Y_N)$ be distributed according to $P_{XY_1\cdot Y_N}$ supported on alphabets $\mathcal{X}$ and $\mathcal{Y}$, and let the shared secret key size be $|\mathcal{X}|$, i.e., $T=|\mathcal{X}|$. %Furthermore, let the local buffer size be $F$ bits, i.e., $C=1$, $L$ be large enough, and $d_i\neq d_j$, $\forall i\neq j$. %and let $M\in\{\frac{N}{K},\frac{2N}{K},\ldots,N\}$. 
	Then, we have
	%\begin{align}
	%\mathbb{L}(P_{XY},|\mathcal{X}|)\leq \!\sum_{x\in\mathcal{X}}\!\!H(\mathcal{C}'|X=x)\!+\!1+\!\ceil{\log (|\mathcal{X}|)},\label{koonwi}
	%\end{align}
	%where $\mathcal{C}'$ is as defined in \eqref{cache1}, and if $|\mathcal{C}'|$ is finite we have
	\begin{align}
	\mathbb{L}(P_{XY},|\mathcal{X}|) &\leq \sum_{i=1}^{K} \ceil{m_i}+\ceil{\log (|\mathcal{X}|)}\label{tt} \\&\leq\sum_{i=1}^{K}\ceil{\log\left(|\mathcal{U}_i|\right)}+\ceil{\log (|\mathcal{X}|)}.\label{koon2wi}
	\end{align}
	where 
	\begin{align}
	|\mathcal{U}_i| &\leq |\mathcal{X}||\mathcal{U}_1|\cdot|\mathcal{U}_{i-1}|(|\mathcal{Y}_{d_i}|-1)+1,\\
	|\mathcal{U}_1| &\leq |\mathcal{X}|(|\mathcal{Y}_{d_1}|-1)+1,
	\end{align}
	and $|\mathcal{Y}_{d_i}|=2^{F}, \ \forall i\in [K]$.
	%Finally, if $X$ is a deterministic function of $\mathcal{C}'$, we have
	%\begin{align}\label{gohwi}
	%\mathbb{L}(P_{XY},|\mathcal{X}|) \leq \ceil{\log(\left(|\mathcal{C}'|-|\mathcal{X}|+1\right))}+\ceil{\log (|\mathcal{X}|)}.
	%\end{align}
\end{theorem}
 \begin{sproof}
 	The complete proof is provided in Appendix~A.
 	%In the placement phase, we use the same scheme as discussed before. In the delivery phase, we use the following strategy.
 	We use a multi-part code construction to achieve both upper bounds. As shown in Fig. \ref{achieve}, the encoder first encodes the private data $X$ using one-time-pad coding \cite[Lemma~1]{kostala2} and sends it over the shared link, which uses $\ceil{\log(|\mathcal{X}|)}$ bits. The rest follows since by assumption the encoder has access to the realization of $X$ and in the one-time-pad coding, the RV added to $X$ is the shared key, which is of size $|\mathcal{X}|$, and as a result the output has uniform distribution.
 	%Next, the server simply let the response $\mathcal{C}'$ be $(Y_{d_1},\ldots,Y_{d_K})$ and sends it to the encoder using blocks of $CF$ bits, i.e., at $i$-th time it sends $Y_{d_i}$. 
 	At each time, the encoder receives $Y_{d_i}$. Then, the encoder produces $U_1^K$ as follows. First, the encoder receives $Y_{d_1}$ and produces $U_1$ based on FRL using $Y\leftarrow Y_{d_1}$ and $X \leftarrow X$. The encoder uses any lossless codes to encode $U_1$ and sends it over the shared link. %Furthermore, it stores it to the public cache. 
 	Note that $I(U_1;X)=0$, $H(Y_{d_1}|U_1,X)=0$, and $|\mathcal{U}_1|\leq |\mathcal{X}|(|\mathcal{Y}_{d_1}|-1)+1$. Next, the encoder receives $Y_{d_2}$ and produces $U_2$ using Lemma \ref{chaghal} with $k=1$ and encode it to send over the shared link. Latter follows since the encoder has access to $U_1$. We have $I(X;U_1,U_2)=0$, $H(Y_{d_2}|U_1,U_2,X)=0$, and $|\mathcal{U}_2|\leq |\mathcal{X}||\mathcal{U}_1|(|\mathcal{Y}_{d_2}|-1)+1$. After receiving the $i$-th demand, i.e., $Y_{d_i}$, the encoder produces $U_i$ using Lemma \ref{chaghal} with $k=i-1$ and $U_1^{i-1}$. We have $I(X;U_1^i)=0$, $H(Y_{d_i}|U_1^i,X)=0$, and $|\mathcal{U}_i|\leq|\mathcal{X}||\mathcal{U}_1|\cdot|\mathcal{U}_{i-1}|(|\mathcal{Y}_{d_i}|-1)+1$. This procedure continuous until encoding the last demanded file $Y_{d_K}$. Note that the leakage from the delivered messages sent over the shared link to the adversary is zero since by construction we have $I(X;U_1^K)=0$. Furthermore, if we choose the randomness used in the one-time-pad coding independent of $(X,U_1^K)$ we have $I(X;U_1^K,\tilde{X})=0$, where $\tilde{X}$ is the output of the one-time-pad coding.

 	%that the server sends over the shared link to satisfy the users$'$ demands \cite{maddah1}. Note that to produce such a $U$ we follow the construction used in Lemma \ref{aghabwi}. Thus, we have
 	%\begin{align}
 	%H(\mathcal{C}'|X,U)&=0,\label{kharkosde}\\
 	%I(U;X)&=0,
 	%\end{align}  
 	%Since we used the construction as in Lemma \ref{FRL}, 
 	%and $U$ also satisfies \eqref{wi}. Thus, we obtain
 	%\begin{align*}
 	%\mathbb{L}(P_{XY},|\mathcal{X}|)\leq \sum_{x\in\mathcal{X}}H(\mathcal{C}'|X=x)+\!1+\!\ceil{\log (|\mathcal{X}|)}.
 	%\end{align*}
 	%To prove \eqref{koon2wi} we use the same coding scheme with the bound in \eqref{prwi}. If $X$ is a deterministic function of $\mathcal{C}'$, we can use the bound in \eqref{provewi} that leads to \eqref{gohwi}. Moreover, for the leakage constraint we note that the randomness of one-time-pad coding is independent of $X$ and the output of the FRL. %Let $Z$ be the compressed output of the EFRL and $\tilde{X}$ be the output of the one-time-pad coding.
 \begin{figure}[h]
 	\centering
 	\includegraphics[scale = .09]{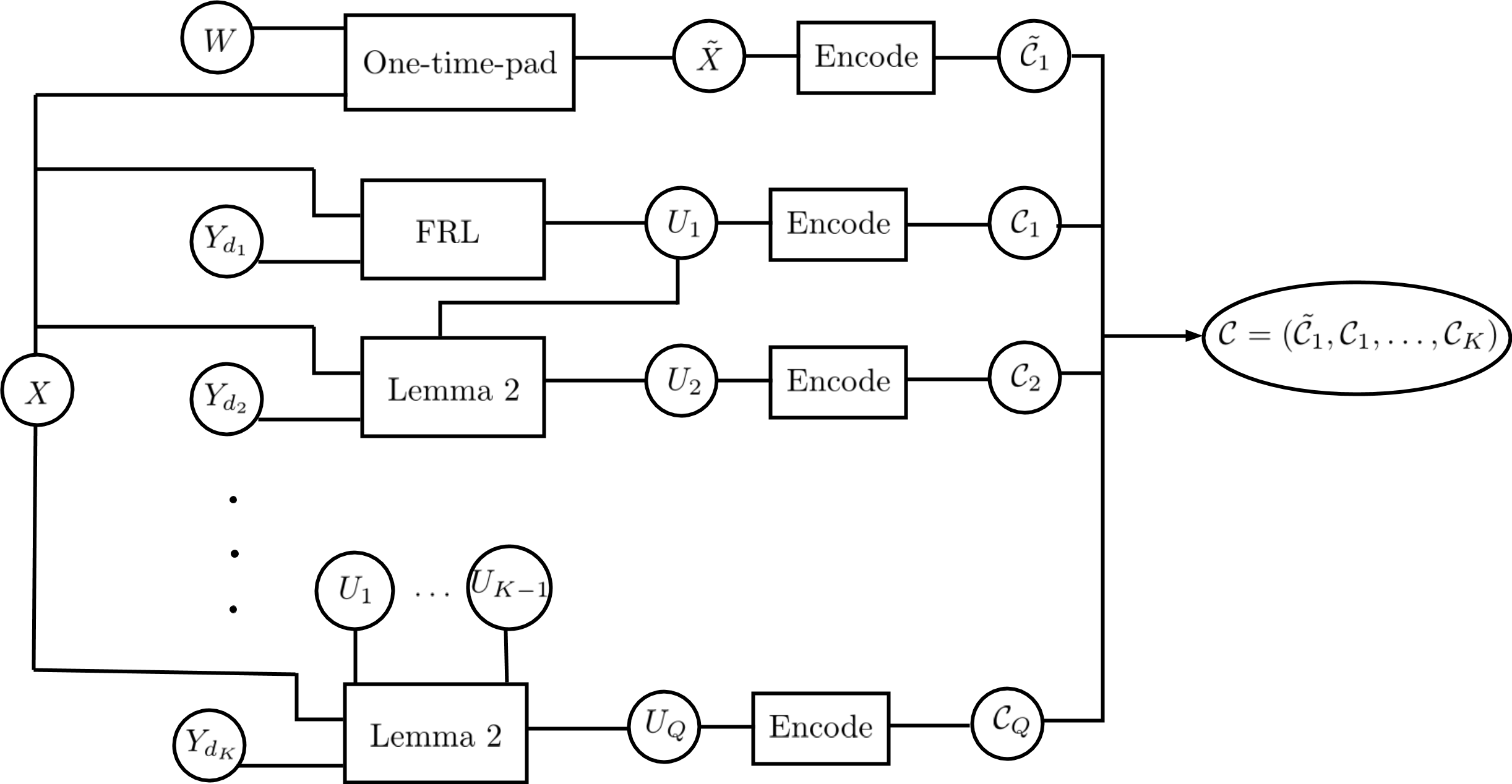}
 	\caption{Sequentially encoding design: illustration of the achievability scheme of Theorem \ref{th1}. Multi-part code construction is used to produce the delivered message, $\mathcal{C}$. The encoder sends $\cal C$ over the link, which is independent of $X$.}
 	%In this work an encoder wants to compress $Y$ which is correlated with $X$ under certain privacy leakage constraints and send it over a channel where an eavesdropper has access to the output of the encoder. The encoder and decoder have an advantage of using shared secret key.}
 	\label{achieve}
 \end{figure}
As shown in Fig. \ref{decode}, at user side, the user, first decodes $X$ using one-time-pad decoder. Then, it decodes $Y_{d_i}$ using $(\mathcal{C}_1,\ldots,\mathcal{C}_i)$ and $X$. %Finally, it decodes $Y_{d_i}$ using local cache $Z_i$ and the response $\mathcal{C}'$. 
Using this scheme the compression rate in \eqref{koon2wi} is achieved. Finally, following the similar approach \eqref{tt} is obtained. 
\begin{figure}[h]
	\centering
	\includegraphics[scale = .09]{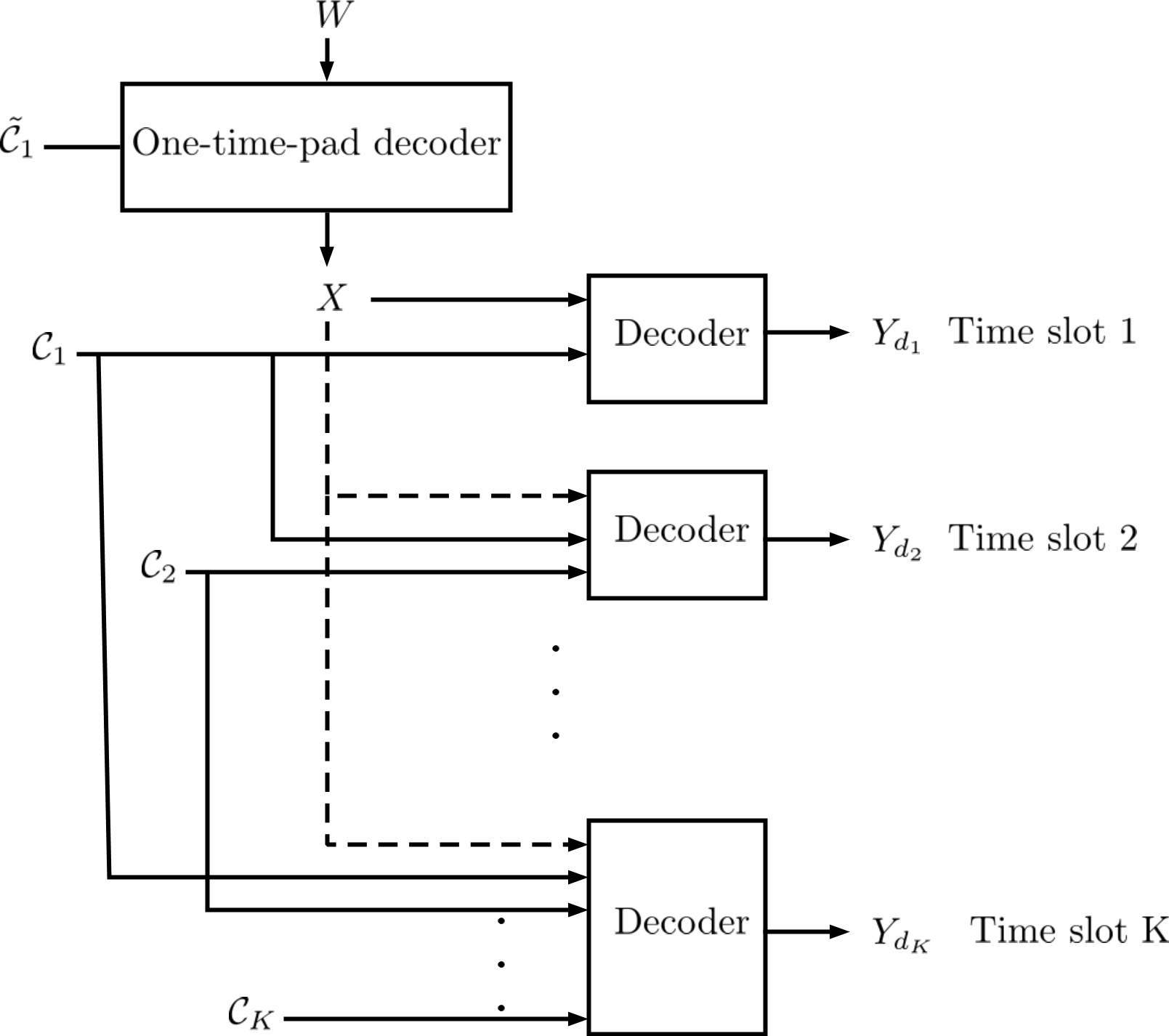}
	\caption{Illustration of the sequentially decoding process of the user for the achievability scheme of Theorem \ref{th1}.}
	%In this work an encoder wants to compress $Y$ which is correlated with $X$ under certain privacy leakage constraints and send it over a channel where an eavesdropper has access to the output of the encoder. The encoder and decoder have an advantage of using shared secret key.}
	\label{decode}
\end{figure}
  \end{sproof}
\\
 Next, we obtain a converse bound (lower bound) on \eqref{main1wi}.
 \begin{theorem}\label{convv}
 	Let the shared key $W$ be independent of $(X,Y)$. For any RVs $(X,Y_1,\ldots,Y_N)$ distributed according to $P_{XY_1\cdot Y_N}$, where $|\mathcal{X}|$ is finite and $|\mathcal{Y}|$ is finite or countably infinite and any shared key size $T\geq 1$ we have
 	\begin{align}\label{lowerr}
 	\mathbb{L}(P_{XY},T)\geq \max_{x\in\mathcal{X}} H(Y_{d_1},\ldots,Y_{d_K}|X=x).
 	\end{align}
 \end{theorem}
\begin{proof}
	Let $U=(U_1,\ldots,U_K)$ be the response of the encoder to satisfy the demands $(d_1,\ldots,d_K)$. Hence, $U$ satisfies
	\begin{align}
	I(U;X)&=0,\\
	H(Y_{d_1},\ldots,Y_{d_K}|W,U)&=0.\label{choo}
	\end{align}
	Noting that \eqref{choo} implies $H(Y_{d_1},\ldots,Y_{d_K}|W,U,X)=0$, thus, by using \cite[Lemma 6]{kostala} we obtain the result.
\end{proof}
Next we study the lower and upper bounds derived in Theorem 1 and Theorem 2 in a numerical example and we obtain an asymptotic constant-factor approximation of the $\mathbb{L}(P_{XY},T)$ when $F$ is large enough.
\begin{example}
	Let $X\sim \text{Bern}(p)\in \{0,1\}$ and $Z_{i}^{j}\sim \text{Bern}(\frac{1}{2})$ be i.i.d RVs for $i\in[F]$ and $j\in[N]$. We also assume that $\{Z_{i}^{j}\}$ is independent of $X$. Furthermore, let $Y_{i}^{j}=Z_{i}^{j} \wedge X$, where $Y_{i}^{j}$ corresponds to the $i$-th bit of the $j$-th file in the database and $\wedge$ stands for the logical ``and'' operator. Let $(Y_{d_1},\ldots,Y_{d_K})$ be the demanded files by the user. We have $H(Y_{d_1},\ldots,Y_{d_K}|X=0)=0$ since $Y_{i}^{j}=0$. Moreover, we have 
	\begin{align}\label{az}
	H(Y_{d_1},\ldots,Y_{d_K}|X=1)&= H(Z_{d_1},\ldots,Z_{d_K}|X=1)\nonumber\\&=H(Z_{d_1},\ldots,Z_{d_K})=KF.
	\end{align}
	Hence, using Theorem 2 we have $\mathbb{L}(P_{XY},T)\geq KF$. Using Theorem 1 and \eqref{az} the ratio between upper and lower bounds can be obtained asymptotically as follows
	\begin{align}\label{zer}
	&\lim_{F\rightarrow \infty}\frac{\sum_{i=1}^{K}\ceil{\log\left(|\mathcal{U}_i|\right)}+\ceil{\log (|\mathcal{X}|)}}{\max_{x\in\mathcal{X}} H(Y_{d_1},\ldots,Y_{d_K}|X=x)}=\nonumber\\ &\lim_{F\rightarrow \infty}\!\!\frac{\log(2^{F+1}\!\!-\!1)\!+\!\log(2(2^{F+1}\!\!-\!1)(2^F\!\!-\!\!1)\!+\!1)\!+\!\ldots\!+\!1}{KF}\!\!=\nonumber\\&\lim_{F\rightarrow \infty}\frac{\log(2^{\frac{K(K+1)}{2}F})}{KF}=\frac{K+1}{2}.
	\end{align}
	Using \eqref{zer}, asymptotically we have
	\begin{align}
	1\leq \lim_{F\rightarrow \infty} \frac{\mathbb{L}(P_{XY},2)}{\max_{x\in\mathcal{X}} H(Y_{d_1},\ldots,Y_{d_K}|X=x)}\leq\frac{K+1}{2}. 
	\end{align} 
	In other words, we have
	\begin{align}
	K\leq \lim_{F\rightarrow \infty} \frac{\mathbb{L}(P_{XY},2)}{F}\leq\frac{K(K+1)}{2}. 
	\end{align} 
	For instance, if the user asks for two files, i.e., $K=2$, we have $2\leq \lim_{F\rightarrow \infty} \frac{\mathbb{L}(P_{XY},2)}{F}\leq 3$.
\end{example}
\subsection{Application: Cache-aided networks with sequential encoder}
\begin{figure}[]
	\centering
	\includegraphics[scale = .11]{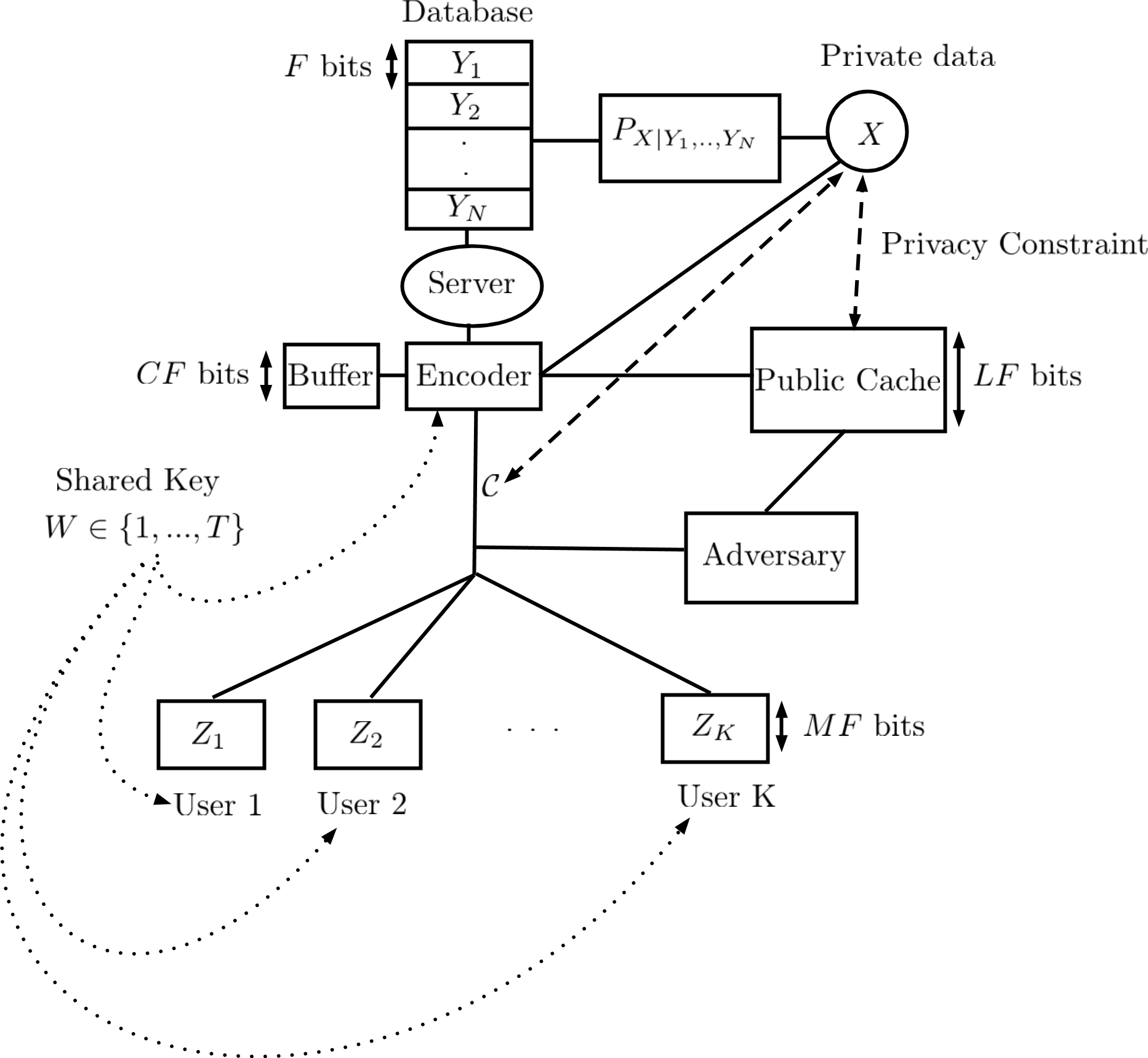}
	\caption{In this work a server wants to send a response over a shared link to satisfy users' demands, but since the database is correlated with the private data existing schemes are not applicable. In the delivery phase, 
		the server sends the response to an encoder using blocks of size $CF$ bits and the encoder hide the information about $X$ using one-time-pad coding and send the rest of message using an extension of Functional Representation Lemma (FRL).}
	%In this work an encoder wants to compress $Y$ which is correlated with $X$ under certain privacy leakage constraints and send it over a channel where an eavesdropper has access to the output of the encoder. The encoder and decoder have an advantage of using shared secret key.}
	\label{wiii}
\end{figure}
We consider the scenario illustrated in Fig. \ref{wiii}, in which a server has access to a database consisting of $N$ files $Y_1,..,Y_N$, where each file, of size $F$ bits, is sampled from the joint distribution $P_{XY_1\cdot Y_N}$, where $X$ denotes the private latent variable.
We assume that the server  does not know the realization of the private variable $X$. The server is connected to 
an encoder equipped with a local buffer of size $CF$ bits and has access to a public cache of size $LF$ bits as well. 
The encoder is connected to $K$ users over a shared link, where user $i$ has access to a local cache memory of size $MF$ bits. Furthermore, we assume that the encoder and the users have access to a shared secret key denoted by $W$, of size $T$. The system works in two phases: the placement and delivery phases, respectively, \cite{maddah1}. In the placement phase, the server and encoder fill the local caches using the database. Let $Z_k$ denote the content of the local cache memory of user $k$, $k\in[K]$ after the placement phase. In the delivery phase, first the users send their demands to the encoder and the server, where $d_k\in[N]$ denotes the demand of user $k$. The server sends a response, denoted by $\mathcal{C}'=(\mathcal{C}'_1,\mathcal{C}'_2,\ldots)$, to the encoder using blocks each of size $CF$ bits to satisfy all the demands. The encoder designs a message $\mathcal{C}_i$ using the $i$-th block of the server's response, i.e., $\mathcal{C}_i'$, and the shared key and send it over the shared link to the users. The encoder can store $\mathcal{C}_i$ in the public cache which is assumed to be not accessible by the users. The public cache enables the encoder to use previous stored messages to design the next message, i.e., the encoder can use $(\mathcal{C}_1,\ldots,\mathcal{C}_{i-1})$ to design message $\mathcal{C}_i$.
We assume that an adversary has access to the shared link and the public cache as well, and uses the delivered message $\mathcal{C}=(\mathcal{C}_1,\mathcal{C}_2,\ldots)$ and the content of public cache denoted by $P$ to extract information about $X$. However, the adversary does not have access to the users' local cache contents or the secret key. %, i.e., it only listens to the channel and receives the response $\cal C$. 
Since the files in the database are all correlated with the private latent variable $X$, the coded caching and delivery techniques introduced in \cite{maddah1} do not satisfy the privacy requirement. The goal of the cache-aided private delivery problem is to find a response $\mathcal{C}$ with minimum possible average length that satisfies a certain privacy constraint and the zero-error decodability constraint of users. Furthermore, the contents of the public cache $P$ must satisfy the privacy constraint. Here, we consider the worst case demand combinations $d=(d_1,..,d_K)$ to construct $\cal C$, and the expectation is taken over the randomness in the database. %In this work, we consider a perfect privacy constraint, i.e., we require $(\mathcal {C},P)$ to be independent of $X$. Let $\hat{Y}_{d_k}$ denote the decoded message of user $k$ using $W$, $\cal C$, and $Z_k$. User $k$ should be able to recover $Y_{d_k}$ reliably, i.e., $\mathbb{P}{\{\hat{Y}_{d_k}\neq Y_{d_k}\}}=0$, $\forall k\in[K]$.
Similarly to \cite{maddah1}, we have $K$ caching functions to be used during the placement phase:
\begin{align}
\theta_k: [|\mathcal{Y}|] \rightarrow [2^{\floor{FM}}],\ \forall k\in[K], 
\end{align} 
such that
\begin{align}
Z_k=\theta_k(Y_1,\ldots,Y_N),\ \forall k\in[K].
\end{align} 
%Let the vector $(Y_{d_1},\ldots,Y_{d_K})$ denote the demands sent by the users at the beginning of the delivery phase, where $(d_1,\ldots,d_K)\in[N]^K$.  
%A variable-length prefix-free code with a shared secret key of size $T$ consists of mappings:
%\begin{align*}
%&(\text{server})\! \ \mathcal{C}': [|\mathcal{Y}|]\times[N]^K \rightarrow \{0,1\}^*\\
%&(\text{encoder})\! \ \mathcal{C}: \{0,1\}^{\floor{CF}}\!\!\times\!\! \{0,1\}^{\floor{LF}}\!\!\times\!\! [T]\!\!\times\!\![N]^K\!\!\! \rightarrow\!\! \{0,\!1\}^*\\
%&(\text{decoder}) \mathcal{D}_k\!: \! \{0,1\}^*\!\!\times\! [T]\!\times\! [2^{\floor{MF}}]\!\times\! [N]^K\!\!\!\rightarrow\! 2^F\!\!\!,\ \! \forall k\!\in\![K].
%\end{align*}
The output of the server $\mathcal{C}'(Y,d_1,\ldots,d_K)$ is the codeword the server sends to the encoder by using blocks of $CF$ bits to satisfy the demands of users $(Y_{d_1},\ldots,Y_{d_K})$. Let $\mathcal{C}'_i$ be the $i$-th block of the codeword $\mathcal {C}'$. 
Due to the limited size of the local buffer $C$, the encoder uses a sequentially coding scheme as follows. 
First, the encoder receives $\mathcal{C}'_1$ and encode it using the shared key. Let the output be $\mathcal{C}_1$. The encoder sends $\mathcal{C}_1$ over the shared link and also stores it to the public cache. Then, the encoder receives the second block $\mathcal{C}'_2$ and encodes it using $W$ and $\mathcal{C}_1$. The output is denoted by $\mathcal{C}_2$ and is sent through the shared link and is stored to the public cache.  
Here, we assume that $L$ is a large quantity. Furthermore, the encoder encodes $i$-th block $\mathcal{C}'_i$ using $W$ and $(\mathcal{C}_1,\ldots,\mathcal{C}_{i-1})$. The encoding scheme continues until the last block of the codeword $\mathcal {C}'$. Let $\mathcal{C}=(\mathcal{C}_1,\mathcal{C}_2,\ldots)$ denote the delivered message consisting of the messages sent over the shared link $\mathcal{C}_i$.  
%The output of the encoder $\mathcal{C}(Y,W,d_1,\ldots,d_K)$ is the codeword the server sends over the shared link in order to satisfy the demands of the users $(Y_{d_1},\ldots,Y_{d_K})$. 
At the user side, user $k$ employs the decoding function $\mathcal{D}_k$ to recover its demand $Y_{d_k}$, i.e., $\hat{Y}_{d_k}=\mathcal{D}_k(Z_k,W,\mathcal{C},d_1,\ldots,d_K)$.
%Since the code is prefix free, no codeword in the image of $\cal C$ is a prefix of another codeword. 
The variable-length code $(\mathcal{C}',\mathcal{C},\mathcal{D}_1,..,\mathcal{D}_K)$ is lossless if for all $k\in[K]$ we have
\begin{align}\label{choon1}
\mathbb{P}(\mathcal{D}_k(\mathcal{C},W,Z_k,d_1,\ldots,d_K)\!=\!Y_{d_k})\!=\!1.
\end{align} 
In the following, we define perfectly private codes. %and $\epsilon$-private codes.
The code $(\mathcal{C}',\mathcal{C},\mathcal{D}_1,\ldots,\mathcal{D}_K)$ is \textit{perfectly private} if
\begin{align}
I(\mathcal{C};X)=0.\label{lashwi1}
\end{align}
%and the code is 
%\textit{$\epsilon$-private} if
%\begin{align}
%I(\mathcal{C};X)=\epsilon.\label{lash1}
%\end{align}
Let $\xi$ be the support of $\mathcal{C}$, where $\xi\subseteq \{0,1\}^*$. For any $c\in\xi$, let $\mathbb{L}(c)$ be the length of the codeword. The lossless code $(\mathcal{C}',\mathcal{C},\mathcal{D}_1,\ldots,\mathcal{D}_K)$ is \textit{$(\alpha,T,d_1,\ldots,d_K)$-variable-length} if 
\begin{align}\label{jojowi1}
\mathbb{E}(\mathbb{L}(\mathcal{C}))\!\leq\! \alpha,\ \forall w\!\in\!\![T]\ \text{and}\ \forall d_1,\ldots,d_K,
\end{align} 
and $(\mathcal{C}',\mathcal{C},\mathcal{D}_1,\ldots,\mathcal{D}_K)$ satisfies \eqref{choon1}.
Finally, let us define the set $\mathcal{H}(\alpha,T,d_1,\ldots,d_K)$ as follows:\\ %and $\mathcal{H}^{\epsilon}(\alpha,T,d_1,\ldots,d_K)$ as follows:\\
$\mathcal{H}(\alpha,T,d_1,\ldots,d_K)\triangleq\{(\mathcal{C}',\mathcal{C},\mathcal{D}_1,\ldots,\mathcal{D}_K): (\mathcal{C}',\mathcal{C},\mathcal{D}_1,\ldots,\mathcal{D}_K)\ \text{is}\ \text{perfectly-private and}\\ (\alpha,T,d_1,\ldots,d_K)\text{-variable-length}  \}$. % and $\mathcal{H}^{\epsilon}(\alpha,T,d_1,\ldots,d_K)\triangleq\{(\mathcal{C}',\mathcal{C},\mathcal{D}_1,\ldots,\mathcal{D}_K): (\mathcal{C}',\mathcal{C},\mathcal{D}_1,\ldots,\mathcal{D}_K)\ \text{is}\ \epsilon\text{-private and}\\ (\alpha,T,d_1,\ldots,d_K)\text{-variable-length}  \}.$
The cache-aided private compression design problem can be stated as follows
\begin{align}
\mathbb{L}(P_{XY_1\cdot Y_N},T,C)&=\!\!\!\!\!\inf_{\begin{array}{c} 
	\substack{(\mathcal{C}',\mathcal{C},\mathcal{D}_1,\ldots,\mathcal{D}_K)\in\mathcal{H}(\alpha,T,d_1,\ldots,d_K)}
	\end{array}}\alpha.\label{main1wi1}%\\
%\mathbb{L}^{\epsilon}(P_{XY_1\cdot Y_N},T,C)&=\!\!\!\!\!\!\!\!\!\!\inf_{\begin{array}{c} 
%	\substack{(\mathcal{C}',\mathcal{C},\mathcal{D}_1,\ldots,\mathcal{D}_K)\in\mathcal{H}^{\epsilon}(\alpha,T,d_1,\ldots,d_K)}
%	\end{array}}\alpha.\label{main2wi}
\end{align} 
In the next theorem, let $\mathcal{C}'\triangleq(C_{\gamma_1},\ldots,C_{\gamma_Q})$ be the code designed by the achievable scheme in \cite[Theorem 1]{maddah1}, where $Q\triangleq\binom{K}{p+1}$, $M\in\{\frac{N}{K},\frac{2N}{K},\ldots,N\}$, and $p\triangleq\frac{KM}{N}$ is an integer.
\begin{theorem}\label{th3}
	%Let RVs $(X,Y)=(X,Y_1,\ldots,Y_N)$ be distributed according to $P_{XY_1\cdot Y_N}$ supported on alphabets $\mathcal{X}$ and $\mathcal{Y}$, and let the shared secret key size be $|\mathcal{X}|$, i.e., $T=|\mathcal{X}|$. Furthermore, 
	Let the local buffer size be $\frac{F}{\binom{K}{p}}$ bits, i.e., $C=\frac{1}{\binom{K}{p}}$, $L$ be large enough, and let $M\in\{\frac{N}{K},\frac{2N}{K},\ldots,N\}$. Then, we have
	%\begin{align}
	%\mathbb{L}(P_{XY},|\mathcal{X}|)\leq \!\sum_{x\in\mathcal{X}}\!\!H(\mathcal{C}'|X=x)\!+\!1+\!\ceil{\log (|\mathcal{X}|)},\label{koonwi}
	%\end{align}
	%where $\mathcal{C}'$ is as defined in \eqref{cache1}, and if $|\mathcal{C}'|$ is finite we have
	\begin{align}
	\mathbb{L}(P_{XY},|\mathcal{X}|,\frac{1}{\binom{K}{p}}) \leq \sum_{i=1}^{Q}\ceil{\log\left(|\mathcal{U}_i|\right)}+\ceil{\log (|\mathcal{X}|)}.\label{koon2wifs}
	\end{align}
	where 
	\begin{align}
	|\mathcal{U}_i| &\leq |\mathcal{X}||\mathcal{U}_1|\cdot|\mathcal{U}_{i-1}|(|C_{\gamma_i}|-1)+1,\\
	|\mathcal{U}_1| &\leq |\mathcal{X}|(|C_{\gamma_1}|-1)+1,
	\end{align}
	and $|C_{\gamma_i}|=2^{CF}=2^{\frac{F}{\binom{K}{p}}}, \ \forall i\in [Q]$.
\end{theorem}
\begin{proof}
	The proof follows similar arguments as Theorem 1. In the placement phase, we use the same scheme as in \cite{maddah1}. In the delivery phase, we use the following strategy.
	We use multi-part code construction to achieve the upper bound. The encoder first encodes the private data $X$ using one-time-pad coding \cite[Lemma~1]{kostala2} and sends it over the shared link, which uses $\ceil{\log(|\mathcal{X}|)}$ bits. %The rest follows since by assumption the encoder has access to the realization of $X$ and in the one-time-pad coding, the RV added to $X$ is the shared key, which is of size $|\mathcal{X}|$, and as a result the output has uniform distribution.
	Next, the server produces the response $\mathcal{C}'$ as in \cite[Theorem 1]{maddah1} and sends it to the encoder using blocks of $CF$ bits. At each time, the encoder receives $C_{\gamma_i}$ since the size of buffer equals to the size of $C_{\gamma_i}$. Then, the encoder produces $U_1^Q$ as follows. First, the encoder receives $C_{\gamma_1}$ and produces $U_1$ based on FRL using $Y\leftarrow C_{\gamma_1}$ and $X \leftarrow X$, where $C_{\gamma_1}$, is the first codeword of the response $\mathcal{C}'$. The encoder uses any lossless codes to encode $U_1$ and sends it over the shared link. Furthermore, the encoder stores it to the public cache. Note that $I(U_1;X)=0$, $H(C_{\gamma_1}|U_1,X)=0$, and $|\mathcal{U}_1|\leq |\mathcal{X}|(|\mathcal{Y}|-1)+1$. Next, the encoder receives $C_{\gamma_2}$ and produces $U_2$ using Lemma \ref{chaghal} with $k=1$ and encode it to send over the shared link. Latter follows since the encoder has access to $U_1$ which is stored in the public cache. We have $I(X;U_1,U_2)=0$, $H(C_{\gamma_2}|U_1,U_2,X)=0$, and $|\mathcal{U}_2|\leq |\mathcal{X}||\mathcal{U}_1|(|\mathcal{Y}|-1)+1$. After receiving the $i$-th codeword of the response $\mathcal{C}'$, i.e., $C_{\gamma_i}$, the encoder produces $U_i$ using Lemma \ref{chaghal} with $k=i-1$ and $U_1^{i-1}$ which are stored in the public cache. We have $I(X;U_1^i)=0$, $H(C_{\gamma_i}|U_1^i,X)=0$, and $|\mathcal{U}_i|\leq|\mathcal{X}||\mathcal{U}_1|\cdot|\mathcal{U}_{i-1}|(|\mathcal{Y}|-1)+1$. This procedure continues until encoding the last codeword of the response $C_{\gamma_Q}$. Note that the leakage from the delivered messages sent over the shared link and public cache to the adversary is zero since by construction we have $I(X;U_1^Q)=0$. Furthermore, if we choose the randomness used in the one-time-pad coding independent of $(X,U_1^Q)$ we have $I(X;U_1^Q,\tilde{X})=0$, where $\tilde{X}$ is the output of the one-time-pad coding. At user side, each user, e.g., user $i$, first decodes $X$ using one-time-pad decoder. Then, it sequentially decodes $\mathcal{C}'$ using $\mathcal{C}$ and $X$. Finally, it decodes $Y_{d_i}$ using $Z_i$ and the response $\mathcal{C}'$. Finally, using this scheme the compression rate in \eqref{koon2wifs} is achieved.
\end{proof}
 %\section{conclusion}
 %We have introduced a cache-aided compression problem with privacy constraint, where the information delivered over the shared link during the delivery phase is independent of an underlying variable $X$ that is correlated with the files in the database that can be requested by the users. We propose an achievable scheme using multi-part code construction that benefits from the  %It has been shown that by using two-part construction coding the achievability scheme can be designed. To do so, we have used 
 %Functional Representation Lemma and an extension of the Functional Representation Lemma to encode the response of the server, which is constructed based on the specific user demands. For the first part of the code we hide the private data using one-time pad coding.
 
\section*{Appendix A}
\subsection*{Proof of Theorem~\ref{th1}:}
	%In the placement phase we use the same scheme as in \cite{maddah1}. In the delivery phase, 
	%we use the following strategy. 
	Let $W$ be the shared secret key with key size $T=|\mathcal{X}|$, which is uniformly distributed over $\{1,,\ldots,T\}=\{1,\ldots,|\mathcal{X}|\}$ and independent of $(X,Y_1,\ldots,Y_N)$. As shown in Fig. \ref{achieve}, first, the private data $X$ is encoded using the shared secret key \cite[Lemma~1]{kostala2}. Thus, we have
\begin{align*}
\tilde{X}=X+W\ \text{mod}\ |\mathcal{X}|.
\end{align*}
Next, we show that $\tilde{X}$ has uniform distribution over $\{1,\ldots,|\mathcal{X}|\}$ and $I(X;\tilde{X})=0$. We have
\begin{align}\label{t}
H(\tilde{X}|X)\!=\!H(X\!+\!W|X)\!=\!H(W|X)\!=\!H(W)\!=\! \log(|\mathcal{X}|).
\end{align}
Furthermore, $H(\tilde{X}|X)\leq H(\tilde{X})$, and combining it with \eqref{t}, we obtain $H(\tilde{X}|X)= H(\tilde{X})=\log(|\mathcal{X}|)$. For encoding $\tilde{X}$ the encoder uses $\ceil{\log(|\mathcal{X}|)}$ bits and sends it over the shared link. We denote the encoded message $\tilde{X}$ by $\tilde{\mathcal{C}}_1$. %Let $\mathcal{C}'$ be the response that the server sends in the delivery phase for the scheme proposed in \cite[Theorem 1]{maddah1}, also given in \eqref{cache1}. 
%Due to the local buffer's size, the server sends $\mathcal{C}'$ using blocks of $CF$ bits, i.e., sends $Y_{d_i}$ at the $i$-th time. Then, the encoder produces $U_1^Q=(U_1,\ldots,U_Q)$ as follows. 
First, the encoder receives $Y_{d_1}$ and produces $U_1$ based on FRL using $Y\leftarrow Y_{d_1}$ and $X \leftarrow X$. Thus, we have
$
I(U_1;X)=0$, $H(Y_{d_1}|U_1,X)=0$, $|\mathcal{U}_1|\leq |\mathcal{X}|(|\mathcal{Y}|-1)+1$.
%\end{align}
The encoder uses any lossless codes to encode $U_1$ and sends it over the shared link. Let $\mathcal{C}_1$ denote the encoded message $U_1$. Furthermore, the encoder stores $\mathcal{C}_1$ to the public cache. Next,
the encoder receives $Y_{d_2}$ and produces $U_2$ based on Lemma \ref{chaghal} with $k=1$ and $U_1$ that is stored in the public cache, encode it to $\mathcal{C}_2$ and sends it over the shared link. Thus, we obtain
$
I(U_1,U_2;X)=0$,  $H(Y_{d_2}|U_2,U_1,X)=0$, and $|\mathcal{U}_2|\leq |\mathcal{X}||\mathcal{U}_1|(|\mathcal{Y}_{d_2}|-1)+1$.
Furthermore, the encoder stores $\mathcal{C}_2$ in the public cache.
After receiving $Y_{d_i}$, the encoder produces $U_i$ using Lemma \ref{chaghal} with $k=i-1$ and $U_1^{i-1}$ which are stored in the public cache.
We have
\begin{align}
I(U_1,\ldots,U_i;X)&=0\label{kos2}, \\ H(Y_{d_i}|U_1,\ldots,U_i,X)&=0,\label{kos1}\\ |\mathcal{U}_i|&\leq |\mathcal{X}||\mathcal{U}_1|\cdot |\mathcal{U}_{i-1}|(|\mathcal{Y}_{d_i}|-1)+1.
\end{align}
This procedure continuous until encoding $Y_{d_K}$ to $\mathcal{C}_K$. By using the construction used in Lemma \ref{chaghal}, we have
\begin{align*}
&H(\tilde{X},U_1^K)=H(\tilde{X})+\sum_i H(U_i)\\ &\leq \sum_{i=1}^{K}\ceil{\log\left(|\mathcal{X}||\mathcal{U}_1|\cdot |\mathcal{U}_{i-1}|(|\mathcal{Y}_{d_i}|-1)+1\right)}+\ceil{\log (|\mathcal{X}|)},
\end{align*}
where we used the fact that $U_1,\ldots,U_K,\tilde{X}$ are jointly independent. Furthermore, we choose $W$ to be independent of $(X,Y_1,\ldots,Y_N,U_1,\ldots,U_K)$ which results in $I(\tilde{X};U_1^K)=0$. %This shows \eqref{koon2wi}.

%We produce $U$ based on the construction proposed in \cite[Lemma~1]{kostala}, where we use $Y\leftarrow \mathcal{C}'$ and $X \leftarrow X$ in the FRL. Thus, 
%\begin{align}
%H(\mathcal{C}'|X,U)&=0,\label{kharkosde1}\\
%I(U;X)&=0.\label{koonnane}
%\end{align} 
%Since the construction of $U$ is based on Lemma~\ref{aghabwi}, we have
%\begin{align}
%H(U)\leq \sum_x H(\mathcal{C}'|X=x)+1.
%\end{align}
%We encode $U$ and $\tilde{X}$ and denote them by $\mathcal{C}_2$ and $\mathcal{C}_1$, which have average lengths at most $\sum_x H(\mathcal{C}'|X=x)+1$ and $\ceil{\log(|\mathcal{X}|)}$, respectively. To encode $U$ we use any traditional lossless codes. Note that both $\mathcal{C}_2$ and $\mathcal{C}_1$ are lossless and variable length codes. Considering $\mathcal{C}=(\mathcal{C}_1,\mathcal{C}_2)$ we have
%\begin{align}
%\mathbb{L}(P_{XY},|\mathcal{X}|)\leq \sum_{x\in\mathcal{X}}H(\mathcal{C}'|X=x)+\!1+\!\ceil{\log (|\mathcal{X}|)}.
%\end{align}
Next, we present the decoding strategy at the user side. Since $W$ is shared, the user decodes $X$ by using $\tilde{\mathcal{C}}_1$ and $W$, i.e., by using one-time-pad decoder. By adding $|\mathcal{X}|-W$ to $\tilde{X}$ we obtain $X$. 
Then, as shown in Fig. \ref{decode}, user $1$ decodes $Y_{d_1}$ using $\mathcal{C}_1$ and $X$. This follows by $H(Y_{d_1}|U_1,X)$. The user sequentially decode $Y_{d_i}$ based on \eqref{kos1}. %After decoding the whole sequence $(C_{\gamma_1},\ldots,C_{\gamma_Q})$, user $i$ uses its local cache $Z_i$ and $\mathcal{C}'$ to decode $Y_{d_i}$, for more details see \cite[Theorem 1]{maddah1}.

%Then, based on \eqref{kharkosde1} user $i$ decodes $\mathcal{C}'$ using $X$ and $U$. Finally, user $i$ decodes $Y_{d_i}$ using $\mathcal{C}'$ and its local cache contents $Z_i$, for more details see \cite[Theorem 1]{maddah1}.
What remains to be shown is the leakage constraint. We choose $W$ independent of $(X,Y_1,\ldots,Y_N,U_1,\ldots,U_K)$. %Let $P$ denote the contents of the public cache. We have
\begin{align*}
I(\mathcal{C};X)&=I(\tilde{\mathcal{C}}_1,\mathcal{C}_1^K;X)=I(\tilde{X},U_1^K;X)\\&=I(U_1^K;X)+I(\tilde{X};X|U_1^K)=I(\tilde{X};X|U_1^K)\\&\stackrel{(a)}{=}H(\tilde{X}|U_1^K)-H(\tilde{X}|X,U_1^K)\\&=H(\tilde{X}|U_1^K)-H(X+W|X,U_1^K)\\&\stackrel{(b)}{=}H(\tilde{X}|U_1^K)-H(W)\\&\stackrel{(c)}{=}H(\tilde{X})-H(W)=\log(|\mathcal{X}|)\!-\!\log(|\mathcal{X}|)\!=\!0
\end{align*}
where (a) follows from \eqref{kos2}; (b) follows since $W$ is independent of $(X,U_1^K)$; and (c) from the independence of $U_1^K$ and $\tilde{X}$. The latter follows since we have
\begin{align*}
0\leq I(\tilde{X};X|U_1^K) &= H(\tilde{X}|U_1^K)-H(W)\\&\stackrel{(i)}{=}H(\tilde{X}|U_1^K)-H(\tilde{X})\leq 0.
\end{align*}
Thus, $\tilde{X}$ and $U_1^K$ are independent. Step (i) follows by the fact that $W$ and $\tilde{X}$ are uniformly distributed over $\{1,\ldots,|\mathcal{X}|\}$, i.e., $H(W)=H(\tilde{X})$. %Thus, the leakage from public cache and delivered message over the shared link to the adversary is zero. %As a summary, if we choose $W$ independent of $(X,U)$ the leakage to the adversary is zero.
	\bibliographystyle{IEEEtran}
	\bibliography{IEEEabrv,IZS}
\end{document}